\documentclass[runningheads]{llncs}
\usepackage{xcolor}
\usepackage{amsmath}
\usepackage{algorithm}
\usepackage{algpseudocode}
\usepackage{accents}
\usepackage{hyperref}
\usepackage{lineno}
\usepackage{graphicx}

\newcommand{\ubar}[1]{\underaccent{\bar}{#1}}
\algrenewcommand\algorithmicrequire{\textbf{Input:}}
\algrenewcommand\algorithmicensure{\textbf{Output:}}

\begin{document}
\title{Computing all-vs-all MEMs in run-length-encoded collections of HiFi reads \thanks{Supported by Academy of Finland Grants 323233 and 339070}}
%
%
\author{Diego D\'iaz-Dom\'inguez \and
Simon J. Puglisi \and Leena Salmela}
\authorrunning{D. D\'iaz-Dominguez et al.}
%
\institute{Department of Computer Science, University of Helsinki, Finland
\email{\{diego.diaz,simon.puglisi,leena.salmela\}@helsinki.fi}}

\maketitle 

\begin{abstract} We describe an algorithm to find maximal exact matches (MEMs) among HiFi reads with homopolymer errors. The main novelty in our work is that we resort to run-length compression to help deal with errors. Our method receives as input a run-length-encoded string collection containing the HiFi reads along with their reverse complements. Subsequently, it splits the encoding into two arrays, one storing the sequence of symbols for equal-symbol runs and another storing the run lengths. The purpose of the split is to get the BWT of the run symbols and reorder their lengths accordingly. We show that this special BWT, as it encodes the HiFi reads and their reverse complements, supports bi-directional queries for the HiFi reads. Then, we propose a variation of the MEM algorithm of Belazzougui et al. (2013) that exploits the run-length encoding and the implicit bi-directional property of our BWT to compute approximate MEMs. Concretely, if the algorithm finds that two substrings, $a_1 \ldots a_p$ and $b_1 \ldots b_p$, have a MEM, then it reports the MEM only if their corresponding length sequences, $\ell^{a}_1 \ldots \ell^{a}_p$ and $\ell^{b}_1 \ldots \ell^{b}_p$, do not differ beyond an input threshold. We use a simple metric to calculate the similarity of the length sequences that we call the {\em run-length excess}. Our technique facilitates the detection of MEMs with homopolymer errors as it does not require dynamic programming to find approximate matches where the only edits are the lengths of the equal-symbol runs. Finally, we present a method that relies on a geometric data structure to report the text occurrences of the MEMs detected by our algorithm.



\end{abstract}
\section{Introduction}

HiFi reads are a new type of DNA sequencing data developed by PacBio~\cite{wenger19ac}. They are long overlapping strings with error rates (mismatches) comparable to those of Illumina data. They have become popular in recent years as their features improve the accuracy of biological analyses~\cite{logsdon2020long}. Still, mapping a collection of HiFi reads against a reference genome or computing suffix-prefix overlaps among them for \emph{de novo} assembly remain important challenges as these tasks require approximate alignments of millions of long strings. Popular tools that address these problems use seed-and-extend algorithms with minimizers as seeds~\cite{li2018minimap2} for the alignments. This technique is a cheap solution that makes the processing of HiFi reads feasible.

An alternative approach is to use \emph{maximal exact matches} (MEMs) as seeds. A MEM is a match $S[a,b]=S'[a',b']$ between two strings $S$ and $S'$ that cannot be extended to the left or to the right without introducing mismatches or reaching an end in one of the sequences. MEMs are preferable over minimizers because they can capture long exact matches between the reads, thus reducing the computational costs of extending the alignments with dynamic programming. However, they are expensive to find in big collections.


A classical solution to detect MEMs among strings of a large collection is to concatenate the strings in one sequence $S$ (separated by sentinel symbols), construct the suffix tree of $S$, and then traverse its topology to find MEMs in linear time~\cite{gusf97}. Still, producing the suffix tree of a massive collection, although linear in time and space, is expensive for practical purposes. Common approaches to deal with the space overhead are sparse suffix trees~\cite{kh09prac,vy13essamem}, hash tables with $k$-mers~\cite{khiste2015mem,gra19cop}, and Bloom filters~\cite{liu2019fast}.

Another way to deal with the space overhead is to find MEMs on top of a compact suffix tree~\cite{sad07comp,ohl10cst}. For instance, Ohlebusch et al.~\cite{oh10comp} described a method that computes MEMs between two strings via matching statistics~\cite{cha94sub}. Their technique requires only one of the strings to be indexed using a compact suffix tree while the other is kept in plain format. Other more recent methods~\cite{ros22mon,bou21pho,ros22fin} follow an approach similar to that of Ohlebusch et al., but they use the r-index~\cite{gagie2020fully} instead of the compact suffix tree.

The problem with the algorithms that rely on matching statistics is that they consider input collections with two strings (one indexed and the other in plain format). It is not clear how to generalize these techniques to compute all-vs-all MEMs in a collection with an arbitrary number of sequences. A simple solution would be to implement classical MEM algorithms on top of the compact suffix tree. Still, producing a full compact suffix tree is expensive for genomic applications as it requires producing a sampled version of the suffix array~\cite{MM93}, the Burrows--Wheeler transform~\cite{bw94}, and the longest common prefix array~\cite{ka01li}. Although it is possible to obtain these composite data structures in linear time and space, in practice, they might require an amount of working memory that is several times the input size. In this regard, Belazzougui et al.~\cite{bel13ver} proposed a MEM algorithm that only uses the bi-directional BWT of the text, although their idea reports the sequences for the MEM, not their occurrences in the text.

Besides the input size, there is another relevant issue when computing MEMs in HiFi data: homopolymer errors. Concretely, if a segment of the DNA being sequenced has an equal-symbol run of length $\ell$, then the PacBio sequencer might spell imprecise copies of that run in the reads that overlap the segment. These copies have a correct\footnote{The symbol correctly represents the nucleotide that was read from the DNA molecule.} DNA symbol, but the value $\ell$ might be incorrect. In general, homopolymer errors shorten the alignment seeds, which means that the pattern matching algorithm will spend more time performing dynamic programming operations to extend the alignments. In this work, we study the problem of finding MEMs in HiFi reads efficiently. Our strategy is to use run-length encoding to remove the homopolymer errors, and then try to filter out the matches between different sequences that, by chance, were compressed to the same run-length encoded string. 

\subsubsection{Our contribution.} We present a set of techniques to compute all-vs-all MEMs in a collection of HiFi reads of $n$ symbols. We build on the MEM algorithm of Belazzougui et al. \cite{bel13ver} that uses the bi-directional BWT, a versatile succinct text representation that uses $2n\log \sigma + o(2n\log \sigma)$ bits of space. Strings in a DNA collection have two complementary sequences that we need to consider for the matches, meaning that we need to create a copy of the input with the complementary strings and merge all in one collection $\mathcal{R}$. We describe a framework that exploits the properties of these DNA complementary sequences to produce an implicit bi-directional BWT for $\mathcal{R}$ without increasing the input size by a factor of 4x. In addition, we define parameters to detect MEMs in a run-length-encoded representation of $\mathcal{R}$. Concretely, we propose the concept of run-length excess, which we use to differentiate homopolymer errors from sporadic matches generated by the run-length compression. Finally, we describe our variation of the algorithm  of Belazzougui et al.~\cite{bel13ver} for computing MEMs using our implicit bi-directional BWT constructed on a run-length-encoded version of $\mathcal{R}$, denoted $\mathcal{R}^{h}$. Let $S$ be a sequence of length $d=|S|$ that has $x$ occurrences in $\mathcal{R}^{h}$, with $l \leq x$ of them having MEMs with other positions of $\mathcal{R}^{h}$. Once our algorithm detects $S$, it can report its MEMs in $O(\sigma^{2}\log \sigma + x^2d)$ time, where $\sigma$ is the alphabet of $\mathcal{R}^{h}$. We also propose an alternative solution that uses a geometric data structure, and report the MEMs of $S$ in $O((x+\sigma)(1+ \log n_h/\log \log n_h)  + l^{2}d)$ time, where $n_h$ is the number of symbols in $\mathcal{R}^{h}$. 

\section{Preliminaries}

\subsubsection{Rank and select data structures.}~\label{sec:rs_dt}
Given a sequence $B[1,n]$ of symbols over the alphabet $\Sigma=[1,\sigma]$, the operation $\textsf{rank}_{a}(B,i)$, with $i \in [1,n]$ and $a\in\Sigma$, returns the number of times the symbol $a$ occurs in the prefix $B[1,i]$. On the other hand, the operation $\textsf{select}_a(B,r)$ returns the position of the $rth$ occurrence of $a$ in $B$. For binary alphabets, $B$ can be represented in $n+o(n)$ bits so that it is possible to solve $\textsf{rank}_{a}$ and $\textsf{select}_{a}$, with $a \in \{0,1\}$, in constant time~\cite{jac89rank,Cla96}. 

\subsubsection{Wavelet trees.}~\label{sec:wt}
Let $S[1,n]$ be a string of length $n$ over the alphabet $\Sigma=[1,\sigma]$. A wavelet tree~\cite{Grossi2003} is a tree data structure $W$ that encodes $S$ in $n\log \sigma + o(n\log \sigma)$ bits of space and supports several queries in $O(\log\sigma)$ time. Among them, the following are of interest for this work:

\begin{itemize}
    \item $\mathsf{access}(W, i)$: retrieves the symbol at position $T[i]$ 
    \item $\mathsf{rank}_{a}(W, i)$: number of symbols $a$ in the prefix $T[1,i]$ 
    \item $\mathsf{select}_{a}(W, r)$: position $j$ where the \emph{rth} symbol $a$ lies in $S$
\end{itemize}

The wavelet tree can also answer more elaborated queries efficiently~\cite{gag12nwt}. From them, the following are relevant:

\begin{itemize}
  \item $\textsf{rangeList}(W, i, j)$ : the list of all triplets $(c, r^{c}_i, r^{c}_j)$ such that $c$ is one of the distinct symbols within $S[i,j]$, $r^{c}_{i}$ is the rank of $c$ in $S[1,i-1]$, and $r^{c}_j$ is the rank of $c$ in $S[1,j]$. 
  \item $\textsf{rangeCount}(W, i,j,l,r)$ : number of symbols $y \in S[i,j]$ such that $l \leq y \leq r$. 
  
\end{itemize}

It is possible to answer $\mathsf{rangeList}$ in $O(u\log \frac{\sigma}{u})$ time, where $u$ is the number of distinct symbols in $S[i,j]$, and $\mathsf{rangeCount}$ in $O(\log \sigma)$ time.

\subsubsection{Suffix arrays and suffix trees.}\label{sec:sa_st}
Consider a string $S[1,n-1]$ over alphabet $\Sigma[2,\sigma]$, and the sentinel symbol $\Sigma[1]=\texttt{\$}$, which we insert at $S[n]$. The \emph{suffix array}~\cite{MM93} of $S$ is a permutation $S\!A[1,n]$ that enumerates the suffixes $S[i,n]$ of $S$ in increasing lexicographic order, $S[S\!A[i],n] < S[S\!A[i+1],n]$, for $i \in [1,n-1]$.

The suffix trie~\cite{fred60tries} is the trie $T$ induced by the suffixes of $S$. For every $S[i,n]$, there is a path $U=v_{1},v_{2},\ldots,v_{p}$ of length $p=n-i+2$ in $T$, where $v_{1}$ is the root and $v_{p}$ is a leaf. Each edge $(v_{j},v_{j+1})$ in $U$ is labeled with a symbol in $\Sigma$, and concatenating the edge labels from $v_{1}$ to $v_{p}$ produces $S[i,n]$. The child nodes of each internal node $v$ are sorted from left to right according to the ranks of the labels in the edges that connect them to $v$. Further, when two or more suffixes of $S$ have the same $j$-prefix, their paths in $T$ share the first $j+1$ nodes.

It is possible to compact $T$ by discarding each unary path $U=v_i,\ldots,v_{j}$ where every node $v_i, v_{i-1},\ldots,v_{j-1}$ has exactly one child. The procedure consists of removing the subpath $U'=v_{i+1},\ldots,v_{j-1}$ and connect $v_{i}$ with $v_{j}$ by an edge labeled with the concatenation of the labels in $U'$. The result is a compact trie of $n$ leaves and less than $n$ internal nodes called the \emph{suffix tree}~\cite{we73li}.

The suffix tree can contain other special edges that connect nodes from different parts of the tree, not necessarily a parent with its children. These edges are called suffix and Weiner links. Let us denote $\mathsf{label}(v)$ the string that labels the path starting at the root and ending at $v$. Two nodes $u$ and $v$ are connected by a suffix link $(u,v)$ if $\mathsf{label}(u)=aW$ and $\mathsf{label}(v)=W$. Similarly, an explicit Weiner link $(u,v)$ labeled $a$ occurs if $\mathsf{label}(u)=W$ and $\mathsf{label}(v)=aW$. A Weiner link is implicit when, for $\mathsf{label}(u)=W$, the sequence $aW$ matches a proper prefix of a node label (i.e., there is no node labeled $aW$). The suffix and Weiner links along with the suffix tree nodes yield another tree called the suffix link tree. 

%


\subsubsection{The Burrows--Wheeler transform.}\label{sec:bwt}

The \emph{Burrows--Wheeler transform} (BWT)~\cite{bw94} is a reversible string transformation that stores in $BWT[i]$ the symbol that precedes the $ith$ suffix of $S$ in lexicographical order, i.e., $BWT[i] = S[S\!A[i]-1]$ (assuming $S[0]=S[n]=\texttt{\$}$). 

The mechanism to revert the transformation is the so-called $\mathsf{LF}$ mapping. Given an input position $BWT[j]$ that maps a symbol $S[i]$, $\mathsf{LF}(j) = j'$ returns the index $j'$ such that $BWT[j']=S[i-1]$ maps the preceding symbol of $S[i]$. Thus, spelling $S$ reduces to continuously applying $\textsf{LF}$ from $BWT[1]$, the symbol to the left of $T[n]=\texttt{\$}$, until reaching $BWT[j]=\texttt{\$}$.

Implementing $\mathsf{LF}$ requires to encode $BWT$ with a data structure that supports $\mathsf{rank}_{a}$. A standard solution is to use the wavelet tree of Section~\ref{sec:wt}, which enables to answer $\mathsf{LF}$ in $O(\log \sigma)$ time. It is also necessary to have an array $C[1,\sigma]$ storing in $C[c]$ the number of symbols in $S$ that are lexicographically smaller than $c$. This enables the implementation of the inverse function for $\mathsf{LF}$ (denoted as $\mathsf{LF}^{-1}$). That is, given the position $BWT[j]$ that maps $S[i]$, $\mathsf{LF}^{-1}(j)=j'$ returns the index $j'$ such that $BWT[j']$ maps $S[i+1]$. 

The BWT also allows to count the number of occurrences of a pattern $P[1,m]$ in $S$ in $\mathcal{O}(m\log \sigma)$ time. The method, called $\mathsf{backwardsearch}$, builds on the observation that if the range $S\!A[s_{j},e_{j}]$ encoding the suffixes of $S$ prefixed by $P[j,m]$ is known, then it is possible to compute the next range $S\!A[s_{j-1},e_{j-1}]$ with the suffixes of $S$ prefixed by $P[j-1,m]$. This computation, or \emph{step}, requires two operations: $s_{j-1} = C[P[j-1]] + \mathsf{rank}_{P[j-1]}(BWT, s_j-1)+1$ and $e_{j-1}=C[P[j-1]] + \mathsf{rank}_{P[j-1]}(BWT, e_j)$. Thus, after $m$ steps of $O(\log \sigma)$ time each, $\mathsf{backwardsearch}$ will find the range $S\!A[s_{1}, e_{1}]$ encoding the suffixes of $S$ prefixed by $P[1,m]$ (provided $P$ exists as substring in $S$).

\subsubsection{Bi-directional BWT.}\label{sec:bibwt}
The bi-directional BWT~\cite{lam09high} of a string $S[1,n]$ is a data structure that maintains the BWT of $S$ and the BWT of the reverse of $S$ (denoted here as $\bar{S}$). Belazzougui et al.~\cite{bel13ver} demonstrated that it is possible to use this representation to visit the internal nodes in the suffix tree $T$ of $S$ in $O(n\log \sigma)$ time. 

The work of Belazzougui et al. exploits the fact that the suffixes of $S$ prefixed by the label of an internal node $v$ in $T$ are stored in a consecutive range $S\!A[s_v, e_v]$, and that $BWT[s_v, e_v]$ encodes the labels for the Weiner links of $v$. 

Let $S\!A_{S}$ and $BWT_{S}$ be the suffix array and BWT for $S$ (respectively). Equivalently, let $S\!A_{\bar{S}}$ and $BWT_{\bar{S}}$ be the suffix array and BWT for $\bar{S}$. For any sequence $X$, Belazzougui et al. maintain two pairs: $(s_X,e_X)$ and $(s_{\bar{X}}, e_{\bar{X}})$, where $S\!A_{S}[s_{X},e_{X}]$ stores the suffixes of $S$ prefixed by $X$ and $S\!A_{\bar{S}}[s_{\bar{X}}, e_{\bar{X}}]$ stores the suffixes of $\bar{S}$ prefixed by $\bar{X}$. They also define a set of primitives for the encoding $(s_X,e_X)$, $(s_{\bar{X}}, e_{\bar{X}})$ of $X$:

\begin{itemize}
\item $\mathsf{isLeftMaximal}(X):$  $1$ if $BWT_{S}[s_{X},e_{X}]$ contains more than one distinct symbol, and $0$ otherwise.
\item $\mathsf{isRightMaximal}(X):$ $1$ if $BWT_{\bar{S}}[s_{\bar{X}},s_{\bar{X}}]$ contains more than one distinct symbol, and $0$ otherwise.
\item $\mathsf{enumerateLeft}(X):$ list of distinct symbols in $BWT_{S}[s_{X},e_{X}]$.
\item $\mathsf{enumerateRight}(X):$ list of distinct symbols in $BWT_{\bar{S}}[s_{\bar{X}},e_{\bar{X}}]$
\item $\mathsf{extendLeft}(X, c):$ list $(i,j),(i',j')$ where $S\!A_{S}[i,j]$ is the range for $cX$ and $S\!A_{\bar{S}}[i',j']$ is the range for $\bar{X}c$
\item $\mathsf{extendRight}(X, c):$ list $(i,j),(i',j')$ where $S\!A_{S}[i,j]$ is the range for $Xc$ and $S\!A_{\bar{S}}[i',j']$ is the range for $c\bar{X}$
\end{itemize}

The key aspect of the bi-directional BWT is that, every time it performs a left or a right extension in $(s_{X}, e_{X})$ (respectively, $(s_{\bar{X}}, e_{\bar{X}})$), it also synchronizes $(s_{\bar{X}}, e_{\bar{X}})$ (respectively, $(s_{X}, e_{X})$). By encoding $BWT_{S}$ and $BWT_{\bar{S}}$ as wavelet trees (Section~\ref{sec:wt}), it is possible to perform $\mathsf{extendLeft}$ and $\mathsf{extendRight}$ in $O(\log \sigma)$ time using a backward search step (Section~\ref{sec:bwt}), and then synchronizing the other range with $\mathsf{rangeCount}$. The functions $\mathsf{enumerateLeft}$ and $\mathsf{enumerateRight}$ take $O(u \log \frac{\sigma}{u})$ time as they are equivalent to $\mathsf{rangeList}$. Finally, both $\mathsf{isLeftMaximal}$ and $\mathsf{isRightMaximal}$ run in $O(\log \sigma)$ time. 

Belazzougui et al. use the primitives above to traverse the suffix link tree and thus visiting the internal nodes of $T$ in $O(n\log \sigma)$ time.

\section{Our contribution}

\subsection{Definitions}\label{sec:def}

We consider the set $\{\texttt{A},\texttt{C},\texttt{G},\texttt{T}\}$ to be the \emph{DNA alphabet}. For practical reasons, we compact it to the set $\Sigma=[2,5]$, and regard $\Sigma[1]=\texttt{\$}$ as a \emph{sentinel} that is lexicographically smaller than any other symbol. Given a string $R$ in $\Sigma^{*}$, we define an \emph{homopolymer} as an equal-symbol run $R[i,j]=(c,\ell)$ of maximal length storing $\ell=j-i+1>1$ consecutive copies of a symbol $c$. Maximal length means that $i=1$ or $R[i-1]\neq c$, and $j=|R|$ or $R[j+1]\neq c$.

We regard the DNA \emph{complement} as a permutation $\pi[1,\sigma]$ that reorders the symbols in $\Sigma$, exchanging $2$ (\texttt{A}) with $5$ (\texttt{T}) and $3$ (\texttt{C}) with $4$ (\texttt{G}). The permutation does not modify $1$ (\texttt{\$}) as it does not represent a nucleotide (i.e., $\pi(1)=1$). The \emph{reverse complement} of $R$, denoted $\hat{R}$, is the string formed by reversing $R$ and replacing every symbol $R[i]$ by its complement $\pi(R[i])$. Given a DNA symbol $a \in \Sigma$, let us define the operator $\underline{a}=\pi(a)$ to denote the DNA complement of $a$. 

The input for our algorithm is a collection $\mathcal{X}=\{R_{1},R_{2}, \ldots, R_{k}\}$ of $k$ HiFi reads over the alphabet $\Sigma$. However, we operate over the expanded collection $\mathcal{R}=\{R_{1}\texttt{\$},\hat{R}_{1}\texttt{\$},\ldots,R_{k}\texttt{\$}, \hat{R}_{k}\texttt{\$}\}$ storing the reads of $\mathcal{X}$ along with their reverse complements, where all the strings have a sentinel appended at the end. $\mathcal{R}$ has $2k$ strings, with a total of $n = \Sigma^{k}_{i=1} 2(|R_{i}|+1)$ symbols. We refer to every possible sequence over the DNA alphabet that label a MEM in $\mathcal{R}$ as a \emph{MEM sequence}.

\subsection{Description of the problem}\label{sec:dp}

Before developing our ideas, we formalize our problem as follows.

\begin{definition}\label{def:prob}
Let $\mathcal{S} = \{S_{1}, S_{2}, \ldots, S_{k}\}$ be a string collection of $k$ strings and $n$ total symbols. The problem of all-vs-all MEMs consists in reporting every possible pair ($S_x[a,b], S_{y}[a',b']$) such that $S_{x},S_{y} \in \mathcal{S}$, $S_{x}\neq S_{y}$, and $S_x[a,b]=S_{y}[a',b']$ is a MEM of length $b-a+1=b'-a'+1 \geq \tau$, where $\tau$ is a parameter.
\end{definition}

HiFi data is usually strand unspecific, meaning that, for any two reads $R_a,R_b \in \mathcal{X}$, there are four possible combinations in which we can have MEMs: $(R_{a}, R_{b})$, $(\hat{R}_{a}, R_{b})$, $(R_{a}, \hat{R}_{b})$, $(\hat{R}_{a}, \hat{R}_{b})$. We can access all such combinations in $\mathcal{R}$, but not in $\mathcal{X}$. Hence, our algorithmic framework solves the problem of Definition~\ref{def:prob} using $\mathcal{R}$ as input. 

\section{Bi-directional BWT and DNA reverse complements}\label{sec:imp_bi_BWT}

In this section, we explain how to exploit the properties of the DNA reverse complements to implement an \emph{implicit} bi-directional BWT for $\mathcal{R}$ that does not require the BWT of the reverse sequences of $\mathcal{R}$ (see Section~\ref{sec:bibwt}). We assume the BWT of $\mathcal{R}$ is the BCR BWT~\cite{bauer13lw}, a variation for string collections. This decision is for technical convenience, and does not affect the output of our framework. We begin by describing the key property in our implicit bi-directional representation: 

\begin{lemma}\label{lemma:rc}
Let $X$ be a string over alphabet $\Sigma$ that appears as a substring in $\mathcal{R}$. Additionally, let the pairs $(s_X, e_X)$ and $(s_{\hat{X}}, e_{\hat{X}})$ be the ranges in $S\!A$ of $\mathcal{R}$ storing all suffixes prefixed by $X$ and $\hat{X}$, respectively. It holds that the sorted sequence of DNA complement symbols in $BWT[s_X, e_X]$ matches the right-context symbols of the occurrences of $\hat{X}$ when sorted in lexicographical order. This relationship applies symmetrically to $BWT[s_{\hat{X}}, e_{\hat{X}}]$ and the sorted occurrences of $X$.

\end{lemma}

\begin{proof}
Assume the symbol $a \in \Sigma$ appears to the left of $p$ occurrences of $Xb$ in $\mathcal{R}$. We know that for each occurrence of $aXb$ in $\mathcal{R}$ there will be also an occurrence of $\ubar{b}\hat{X}\ubar{a}$ as we enforce that property by including the DNA reverse complements of the original reads (collection $\mathcal{X}$ of Section~\ref{sec:def}). As a result, $BWT[s_{Xb}, e_{Xb}]$ will contain $p$ copies of $a$ and $BWT[s_{\hat{X}\ubar{a}}, e_{\hat{X}\ubar{a}}]$ will contain $p$ copies of $\ubar{b}$.
\end{proof}

We will use Lemma~\ref{lemma:rc} to implement the functions $\mathsf{enumerateRight}$, $\mathsf{extendRight}$ and $\mathsf{isRightMaximal}$ (Section~\ref{sec:bibwt}) on top of the BWT of the text. Unlike the technique of Belazzougui et al., we synchronize the pairs $(s_X, e_X), (s_{\hat{X}}, e_{\hat{X}})$. Another difference is that both pairs $(s_X, e_X), (s_{\hat{X}}, e_{\hat{X}})$ map to the suffix array of the text. In the original version, the second pair maps to the suffix array of the reverse text. To implement the functions above, we need to update both pairs every time we perform $\mathsf{extendLeft}$ and $\mathsf{extendRight}$.

Belazzougui et al. implement $\mathsf{extendLeft}(X, c)$ by performing a backward search step over $BWT[s_X, e_X]$ using the symbol $c$. The result of this operation is the suffix array range for $cX$. To modify $(s_{\hat{X}}, e_{\hat{X}})$ so it maps to the suffix array range for $\hat{X}\ubar{c}$, we sum the frequencies of the distinct symbols within $BWT[s_X, e_X]$ whose DNA reverse complements are lexicographically smaller than $\ubar{c}$. This operation comes directly from Lemma~\ref{lemma:rc}. Assume the sum is $y$ and that the frequency of $c$ in $BWT[s_X, e_X]$ is $z$, then we compute $s_{\hat{X}\ubar{c}} = s_{\hat{X}} + y$ and $e_{\hat{X}\ubar{c}} = s_{\hat{X}} + y + z$. We use a special form of $\mathsf{rangeCount}$ to get the value for $y$. If $c < \ubar{c}$, then we will use $y = \mathsf{rangeCount}(BWT, s_X, e_X, c+1, \sigma)$. In the other case, $c> \ubar{c}$, we use $\mathsf{rangeCount}(BWT, s_X, e_X, 1, c-1)$. The rationale for computing $\mathsf{rangeCount}$ comes from the relationship between complementary nucleotides in the permutation $\pi$ of Section~\ref{sec:def}. The operation $\mathsf{extendRight}(X, c)$ is analogous; we perform the $\mathsf{backwardsearch}$ step over $BWT[s_{\hat{X}}, e_{\hat{X}}]$ using $\ubar{c}$ as input, and then we count the number of symbols that are lexicographically smaller than $c$.

The functions $\mathsf{enumerateRight}(X)$ and $\mathsf{isRightMaximal}(X)$ are implemented with minor changes. The only caveat is that, when we use $\mathsf{enumerateRight}$, we need to spell the DNA reverse complements of the symbols returned by $\mathsf{rangeList}$. 

\begin{corollary}
Given a collection $\mathcal{X}$ of DNA sequences and its expanded version $\mathcal{R}$ that contains the strings of $\mathcal{X}$ along their reverse complement sequences, we can construct an implicit bi-directional BWT index that does not require the BWT of the reverse of $\mathcal{R}$ and that answers the queries $\mathsf{enumerateRight}$, $\mathsf{extendRight}$ and $\mathsf{isRightMaximal}$ in $O(u \log \frac{\sigma}{u})$ and $O(\log \sigma)$ time, respectively, where $u$ is the number of distinct symbols within the input range for $\mathsf{extendRight}$.  
\end{corollary}

Observe the BWT for $\mathcal{R}$ is implicitly bi-directional as the DNA reverse complements are just the reverse strings with their symbols permuted according to $\pi$ (see Definitions). However, in the case of $\mathcal{R}$, both BWTs are merged in a single representation. Producing a standard bi-directional BWT would increase the size of $\mathcal{X}$ by a factor of 4. In real applications where the data is a multiset of DNA sequencing reads, we have to transform $\mathcal{X}$ into $\mathcal{R}$ regardless if we construct a bi-directional BWT as the reads are strand-unspecific (see Section~\ref{sec:dp}). 

\subsubsection{Contraction operations in the implicit bi-directional BWT.} Given a range $S\!A[i,j]$ of suffixes prefixed by a string $X$, and a parameter $w\leq |X|$, a \emph{contraction} operation returns the range $i'\geq i,j \leq j'$ in $S\!A$ storing the suffixes of the text prefixed by $X[1,w]$. It is possible to solve this query efficiently with either the wavelet tree of the LCP or with a compact data structure that encodes the suffix tree's topology. The problem with those solutions is that we have to deal with the overhead of constructing and storing those representations. We describe how to use our implicit bi-directional BWT to visit the ancestors of a node $v$ in the suffix tree in $O(|\mathsf{label}(v)|\log \sigma)$ time to solve contraction operations. This idea is slower than using the LCP or the suffix tree's topology, but it does not require extra space, and it is faster than the quadratic cost of using a regular BWT. Our technique is a byproduct of our framework, and it is of independent interest. The inputs for the ancestors' traversal are the range $S\!A[s_v, e_v]$ for $v$, and its string depth $d=|\mathsf{label}(v)|$. The procedure is as follows: starting from $BWT[s_v]$, we perform $d$ $\mathsf{LF}^{-1}$ operations to spell $\mathsf{label}(v)$. Simultaneously as we spell the sequence, we also perform backward search steps using the DNA complement of the symbols we obtain with $\mathsf{LF}^{-1}$. We use Lemma~\ref{lemma:rc} to keep the ranges of the backward search steps synchronised with the ranges for the distinct prefixes of $\mathsf{label}(v)$. Recall that $\mathsf{backwardsearch}$ consumes the input from right to left. In our case, this input is a sequence $W$ that matches the DNA reverse complement of $\mathsf{label}(v)$. Thus, by Lemma~\ref{lemma:rc}, we know that visiting the $S\!A$ ranges for the suffixes of $W$ is equivalent to visit the $S\!A$ ranges for the prefixes of $\mathsf{label}(v)$. Finally, each time we obtain a new range $S\!A[i',j']$ with the backward search step, we use $\mathsf{isLeftmaximal}$ to check if $BWT[i',j']$ is unary. If that is the case, then we report the synchronized range of $S\!A[i',j']$ as an ancestor of $v$. The rationale is that if $W$ is left-maximal, then $\hat{W}=\mathsf{label}(v)[1, |W|]$ is right-maximal too, and hence, its sequence is the label of an ancestor of $v$ in the suffix tree.

\subsection{Homopolymer errors and MEM sequences}\label{sec:hom_mem}


A MEM algorithm that runs on top of the suffix tree of $\mathcal{R}$ is unlikely to report all the real\footnote{Those we would obtain in a collection with no homopolymer errors.} matches if the input collection is HiFi data. The difficulty is that some of the MEMs are ``masked'' in the suffix tree. More specifically, suppose we have two nodes $v$ and $u$, with $\mathsf{label}(v)\neq \mathsf{label}(u)$. It might happen that, by removing or adding copies of symbols in the equal-symbol runs of $\mathsf{label}(u)$, we can produce $\mathsf{label}(v)$. If those edits are small enough for the PacBio machine to produce them during the sequencing process, then it is plausible to assume that $\mathsf{label}(u)$ is an homopolymer error of $\mathsf{label}(v)$. This situation becomes even more likely if $\mathsf{label}(u)$ is long and its frequency is low in the collection. 

Looking for all the possible suffix tree nodes that only have small differences in the length of homopolymer runs similar to $v$ and $u$
could be expensive. A simple workaround is to run-length compress $\mathcal{R}$ and execute the suffix-tree-based MEM algorithm with that as input. Now the problem is that we can report false positive MEMs between different sequences that have the same run-length representation but that are not homopolymer errors. Fortunately, filtering those false positive is not so difficult. Before explaining our idea, we formally define the notion of equivalence between sequences. 


\begin{definition}~\label{def:eq_seqs}
Let $A$ be a string whose run-length encoding is the sequence of pairs $A=(a_1,\ell_{1}),(a_{2},\ell_{2}), \ldots, (a_{p},\ell_{p})$, where $a_{i}$ is the symbol of the $ith$ equal-symbol run, and $\ell_{i}\geq 1$ is its length. Additionally, let the operator $\mathsf{rlc}(A)=a_1,a_2,\ldots,a_{p}$ denote the sequence of run heads for $A$. We say that two strings $A$ and $B$ are equivalent iff $\mathsf{rlc}(A) = \mathsf{rlc}(B)$.
\end{definition}

We use equivalent sequences (Definition~\ref{def:eq_seqs}) to define a filtering parameter to discard false positive MEMs. We call this parameter the $\emph{run-length excess}$:

\begin{definition}
Let $A$ and $B$ be two distinct strings with $\mathsf{rlc}(A)=\mathsf{rlc}(B)$. Additionally, let the pair sequences $A=(x_{1},\ell^{a}_{1}), (x_{2}, \ell^{a}_{2}),\ldots,(x_{p},\ell^{a}_{p})$ and $B=(x_{1},\ell^{b}_{1}), (x_{2}, \ell^{b}_{2}),\ldots,(x_{p},\ell^{b}_{p})$ be the run-length encoding for $A$ and $B$, respectively. Each $x_{i} \in \Sigma$ is the $ith$ run head, and $\ell^{a}_{i}, \ell^{b}_{i} \geq 1$ are the lengths for $x_{i}$ in $A$ and $B$, respectively. Now consider the string $E=|\ell^{a}_{1}-\ell^{b}_{1}|,\ldots,|\ell^{a}_{n}-\ell^{b}_{n}|$ storing the absolute differences between the run lengths of $A$ and $B$. We define the run-length excess as $\mathsf{rlexcess}(A,B)=\max(E[1],E[2],\ldots,E[n])$. 
\end{definition}

Intuitively, equivalent sequences that have a high run-length excess are unlikely to have a masked MEM. The reason is because, although the PacBio sequencing process makes mistakes estimating the lengths of the equal-symbol runs, the error in the estimation is unlikely to be high.

Now that we have a framework to detect MEMs in run-length-compressed space, we construct a new collection $\mathcal{R}^{h}$ of $n_{h}\leq n$ symbols encoding the same strings of $\mathcal{R}$ but with their homopolymers compacted. Namely, every equal-symbol run $R_{u}[i,j]=(c,\ell)$ of maximal length $\ell>1$ in $\mathcal{R}$ is represented with a special metasymbol $c^{*} \notin \Sigma$ in $\mathcal{R}^{h}$. We store the $\ell$ values in another list $H$, sorted as their respective homopolymers occur in $\mathcal{R}$. Each element of $\Sigma$ has its own metasymbol, including the sentinel. We reorder the alphabet $\Sigma \cup \Sigma^{h}$ of $\mathcal{R}^{h}$ to the set $\{\texttt{\$}, \texttt{A}, \texttt{A}^{*}, \texttt{C}, \texttt{C}^{*}, \texttt{G}^{*}, \texttt{G}, \texttt{T}^{*}, \texttt{T}, \texttt{\$}^{*}\}$, which we map to its compact version $\Sigma^{hp} = [1,10]$. This reordering will facilitate the synchronization of ranges when we perform $\mathsf{extendLeft}$ or $\mathsf{extendRight}$ in our implicit bi-directional BWT. 

Recall from Section~\ref{sec:imp_bi_BWT} that, when we call the operation $\mathsf{extendLeft}(X, c)$ (respectively, $\mathsf{extendRight}(X, c)$), we need to perform $\mathsf{rangeCount}(BWT, s_{X}, e_{X})$ to get the number of symbols within $BWT[s_{X}, e_{X}]$ whose DNA complements are smaller than $c$. For this counting operation to serve to synchronize $BWT[s_{\hat{X}}, e_{\hat{X}}]$ in constant time, we need the BWT alphabet to be symmetric. Concretely, the permutation $\pi$ for the DNA complements has to exchange $\Sigma^{hp}[1]$ with $\Sigma^{hp}[\sigma]$, $\Sigma^{hp}[2]$ with $\Sigma^{hp}[\sigma-1]$, and so on. This is the reason why the sentinel has a metasymbol too, even though there are no sentinel homopolymers in $\mathcal{R}$. Additionally, we define a function $g: \Sigma^{hp} \rightarrow \Sigma$ to map metasymbols back to their nucleotides in $\Sigma$. When the input for $g$ is not a metasymbol, $g$ returns the nucleotide itself.

The next step is to run the suffix-tree-based algorithm to solve the all-vs-all MEM problem of Definition~\ref{def:prob} (see Section~\ref{sec:bibwt}) using $\mathcal{R}^{h}$ as input. However, we add one extra step. For each candidate MEM $(R_{a}[i,j], R_{b}[i',j'])$, with $R_{a}, R_{b} \in \mathcal{R}^{h}$, reported by the algorithm, we check if the run-length excess between $R_{a}[i,j]$ and $R_{b}[i,j]$ is below some minimum threshold $e$. If that is not the case, then we discard that pair as a MEM. We can easily compute the run-length excess value using the suffix array of $\mathcal{R}^{h}$ and the vector $H$. If the MEM algorithm detects that an internal node $v$ of the suffix tree encodes a list of MEMs, then we use the suffix array of $\mathcal{R}^{h}$ to access the text positions $\mathsf{label}(v)$. Subsequently, we map those positions to $H$ to get the lengths of the distinct variations of $\mathsf{label}(v)$ on the text, and thus compute excess among them.

\subsection{Computing MEMs in compressed space}\label{sec:gst_mem}

We now have all the elements to solve Problem~\ref{def:prob} in run-length-compressed
space using our implicit bi-directional BWT. Our input is the BWT of $\mathcal{R}^{h}$ (encoded as a wavelet tree $BWT$), the array $H$ storing the lengths of the homopolymers in the HiFi reads, and the parameters $\tau$ and $e$ for, respectively, the minimum MEM length and the maximum run-length excess (see Section~\ref{sec:hom_mem}).

We resort to the algorithm of Belazzougui et al.~\cite{bel13ver} to visit the internal nodes in the suffix tree $T$ of $\mathcal{R}^{h}$ in $O(n_{h}\log |\Sigma^{hp}|)$ time, with $n_{h}=\Sigma_{1}^{|\mathcal{R}_{h}|} |R_{i}|$ (see Section~\ref{sec:sa_st}). The advantage of their method is that we can use backward search operations over $BWT$ to navigate $T$ without visiting its edge labels (i.e., unary paths in the suffix trie of $\mathcal{R}^{h}$). Algorithm~\ref{algo:st} describes the procedure.

Each internal node $v$ of $T$ with more than one Weiner link (i.e., $BWT[s_v, e_v]$ is not unary) encodes a group of MEMs. This property holds because $\mathsf{label}(v)$ has more than one left-context symbol and more than one right-context symbol in the text. Thus, any possible combination of strings $a{\cdot}\mathsf{label}(v){\cdot}b$ and $y{\cdot}\mathsf{label}(v){\cdot}z$ we can decode from $v$, with $a,b,y,z \in \Sigma^{hp}$, $a\neq y$, and $b\neq z$, corresponds to a MEM sequence (see Definitions). The sequences we obtain from $v$ can have multiple occurrences in $\mathcal{R}^{h}$, and we need to report all of them. However, some of them might be false positives. For instance, the pair of text positions conforming a MEM are in the same string, or in strings that are DNA reverse complements of each other. We filter those cases as they are artefacts in our model.

When we visit a node $v$ with more than one Weiner link during the traversal of $T$, we access its MEM sequences as follows: we use $\mathsf{enumerateRight}$ and $\mathsf{extendRight}$ to compute every range $S\!A[s_u, e_u]$, with $s_{v}\leq s_u \leq e_u\leq e_v$, encoding a child $u$ of $v$. Then, over each $S\!A[s_u, e_u]$, we perform $\mathsf{enumerateLeft}$ and $\mathsf{extendLeft}$ to compute every range $S\!A[s^{c}_u, e^{c}_u]$ encoding a Weiner link $c$ of $u$. This procedure yields a set $\mathcal{M} = \{\mathcal{I}_{1}, \mathcal{I}_{2},\ldots, \mathcal{I}_{p}\}$, where $p$ is the number of children of $v$, and $\mathcal{I}_{q}$, with $q \in [1, p]$, is the set of ranges in $S\!A$ for the Weiner links of the $qth$ child of $v$ (from left to right).

The next step is to report the text positions of the MEM sequences encoded by $\mathcal{M}$. For this purpose, we consider the list of pairs $\{(\mathcal{I}_{e}, \mathcal{I}_{g})\ |\ \mathcal{I}_{e},\mathcal{I}_g \in \mathcal{M}\ \textrm{and}\ \mathcal{I}_e \neq \mathcal{I}_g\}$. Every element $(S\!A[i,j],S\!A[i',j']) \in \mathcal{I}_{e} \times \mathcal{I}_{g}$ is a pair of ranges such that $S\!A[i,j]$ stores the suffixes of $\mathcal{R}^{h}$ prefixed by a label $a{\cdot}\mathsf{label}(v){\cdot}b$ and $S\!A[i',j']$ stores the suffixes of $\mathcal{R}^{h}$ prefixed by another label $y{\cdot}\mathsf{label}(v){\cdot}z$. We know that $b$ and $z$ are different as they come from different children of $v$. However, the symbols $a$ and $y$ might be equal, which means $\mathsf{label}(v)$ is not a MEM sequence when we match $a{\cdot}\mathsf{label}(v){\cdot}b$ and $y{\cdot}\mathsf{label}(v){\cdot}z$. We can find out this information easily: if $S\!A[i,j]$ and $S\!A[i',j']$ come from different buckets\footnote{The $bth$ bucket of $S\!A$ is the range containing all suffixes prefixed by symbol $b \in \Sigma$.}, then $a\neq y$.  If that is the case, we have to report the MEMs associated to $(S\!A[i,j],S\!A[i',j'])$. For doing so, we first get the string depth $d=|\mathsf{label}(v)|$ of $v$. Then, we regard $X=\{i,\ldots,j\}$ and $O=\{i',\ldots,j'\}$ as two different sequences of consecutive indexes in $S\!A$, and iterate over their Cartesian product $X \times O$. When we access a pair ($S\!A[x]$, $S\!A[o]$), with $(x,o) \in X \times O$, we compute the run-length excess $e'$ between $\mathcal{R}^{h}[S\!A[x]+1, S\!A[x]+d]$ and $\mathcal{R}^{h}[S\!A[o]+1, S\!A[o]+d]$ as described in Section~\ref{sec:hom_mem}, and discard the MEM in $(S\!A[x],S\!A[o])$ if $e' \geq e$. We also discard it if $S\!A[x]$ and $S\!A[o]$ map the same string or map different strings that are reverse complements between each other. This procedure is described in Algorithm~\ref{alg:mem}.

\begin{theorem}
Let $\mathcal{R}^{h}$ be the run-length encoded collection of HiFi reads, with an alphabet of $\sigma_{h} = |\Sigma^{hp}|$ symbols. Additionally, let $v$ be an internal node in the suffix tree of $\mathcal{R}^{h}$ that has more than one Weiner link. The string depth of $v$ is $d = |\mathsf{label}(v)|$ and its associated range $S\!A[i,j]$ has length $x=j-i+1$. We can compute all the MEMs encoded by $v$ in $O(\sigma_{h}^{2}\log \sigma_{h} + x^2d)$ time. 
\end{theorem}

\begin{proof}
We first compute the ranges for the children of $v$ with the operations $\mathsf{enumerateRight}$ and $\mathsf{extendRight}$. These two operations take $O(\sigma_{h} \log \sigma_{h})$ time. Then, for every child, we compute its Weiner links. The node $v$ has up to $\sigma_{h}$ children, each child has up to $\sigma_{h}$ Weiner links, and to compute each of these takes $\log \sigma_{h}$ time via $\mathsf{extendLeft}$, making $O(\sigma_{h}^{2}\log \sigma_{h})$ time in total. The number of suffixes of $\mathcal{R}^{h}$ in $\mathcal{M}$ is $x$, and the total number of suffix pairs we visit during the scans of the Cartesian products between sets of $\mathcal{M}$ is bound by $x^{2}$. Each time we visit a pair of suffixes, computing the run-length excess between them takes us $O(d)$ time. Thus, the time for reporting the MEMs from $v$ is $O(\sigma_{h}^{2}\log \sigma_{h} + x^2d)$.\qed
\end{proof}

\subsection{Improving the time complexity for reporting MEMs}



We can think of the problem of reporting MEMs from $v$ as two-dimensional sorting. We need the occurrences of $\mathsf{label}(v)$ to be sorted by their left and right contexts at the same time (the dimensions) to report the MEMs from $v$ efficiently. We can implement this idea using a grid $\mathcal{G}$ with dimensions $n_{h} \times n_{h}$. We (logically) label the rows of $\mathcal{G}$ with the suffixes of $\mathcal{R}^{h}$ sorted in lexicographical order, and do the same with the columns. We then store the values of $S\!A$ in the grid cells, with the (row,column) coordinate for each $S\!A[j]$ being $(j, \mathsf{LF}(j))$. We encode $\mathcal{G}$ with the data structure of Chan et al.~\cite{chan2011or} that increases the space by $O(n_{h}\log n_{h}) + o(n_{h}\log n_{h})$ bits and allows reporting of the $occ$ points in the area $[x_1,x_2],[y_1,y_2]$ of $\mathcal{G}$ in $O((occ+1)(1+\log n_h / \log\log n_h))$ time. In exchange, we no longer require $S\!A$.

The procedure to report MEMs is now as follows. When we reach $v$ during the suffix tree traversal, we perform $\mathsf{extendLeft}$ with each of $v$'s Weiner links. This produces a list $\mathcal{L}$ of up to $\sigma_{h}$ non-overlapping ranges in $S\!A$. We then create another list $\mathcal{Q}$ with the ranges obtained by following $v$'s children. Notice that the ranges of $\mathcal{Q}$ are a partition of the range $[i,j]$ in $S\!A$ for $\mathsf{label}(v)$. For every $[l_1, l_2] \in \mathcal{L}$, we extract the points in $\mathcal{G}$ in the area $[l_1,l_2],[i,j]$, and partition the result into subsets according to $\mathcal{Q}$. The partition is simple as the points can be reported in increasing order of the $y$ coordinates (range $[i,j]$). The idea is to generate a list $\mathcal{I}=\{I_{1},I_{2},\ldots,I_{x}\}$ of at most $\sigma_{h}^{2}$ elements, where each element is a point set for an area $[l_1, l_2],[q_1,q_2] \in \mathcal{L} \times \mathcal{Q}$. Finally, we scan all possible distinct pairs of $\mathcal{I}$ that yield MEMs, processing suffixes as in lines 18-23 of Algorithm~\ref{alg:mem}. Let $I_{i}, I_{j} \in \mathcal{I}$ be two point sets, extracted from the areas $[l_1,l_2],[q_1, q_2]$ and $[l'_1, l'_2],[q'_1, q'_2]$ of $\mathcal{G}$, respectively. The points of $I_{i}$ will have MEMs with the points of $I_{j}$ if $[l_1, l_2] \neq [l'_1, l'_2]$ and $[q_1,q_2]\neq[q'_1,q'_2]$. See Figure~\ref{fig:grid}. 

\begin{corollary}
By replacing $S\!A$ with the grid of~Chan et al.~\cite{chan2011or}, reporting the MEMs associated with internal node $v$ of the suffix tree of $\mathcal{R}^{h}$ takes $O((x+\sigma)(1+ \log n_h/\log \log n_h)  + l^{2}d)$ time, where $x$ is the number of occurrences of $\mathsf{label}(v)$ in $\mathcal{R}^{h}$, $l\leq x$ is the number of those occurrences that have MEMs, and $d=\mathsf{label}(v)$.
\end{corollary}

\section{Concluding remarks}

We presented a framework to compute all-vs-all MEMs in a collection of run-length encoded HiFi reads. Our techniques can be adapted to other types of collections with properties similar to that of HiFi data (e.g., Nanopore sequencing data, proteins, Phred scores, among others). The larger alphabet of proteins and Phred scores make our MEM reporting algorithm that uses the geometric data structure more relevant (as it avoids the $\sigma^{2}$ complexity of our first method). We are also applying these techniques to {\em de novo} assembly of HiFi reads. 

\bibliography{references}
\appendix
\clearpage\section*{Appendix}

\begin{algorithm}[!htp]
\caption{Computing MEMs in one traversal of the suffix tree $T$ of $\mathcal{R}^{h}$. Arrays $BWT$, $S\!A$, and $H$ are implicit in the pseudo-code. Each node $v \in T$ is encoded by the pair $(v, d)$, where $v=(i, j),(i',j')$ are the ranges in $S\!A$ for $\mathsf{label}(v)$ and its DNA reverse complement $\mathsf{label}(\hat{v})$, and $d$ is the string depth.}
\begin{algorithmic}[1]
\Require Suffix tree $T$ for $\mathcal{R}^{h}$ encoded by the implicit bi-directional BWT.
\Ensure MEMs as described in Definition~\ref{def:prob}.
\State $S \gets \emptyset$ \Comment{Empty stack}
\State $r \gets (1, n+1),(1,n+1)$ \Comment{The root of $T$}
\State $\mathsf{push}(S, (r ,0))$ 
\While{$S \neq \emptyset$}

\State $(v, d) \gets \mathsf{top}(S)$ \Comment{Extract suffix tree node $v$ from the top of the stack}
\State $\mathsf{pop}(S)$

\If{$d \geq \tau$ and $\mathsf{isLeftMaximal}(v)$ and $\mathsf{isRightMaximal}(v)$}
\State $\mathsf{repMEM}(v, e, d)$
\EndIf

\For{$c \in \mathsf{enumerateLeft}(v)$}\Comment{Continue visiting other suffix tree nodes}
\State $u \gets \mathsf{extendLeft}(v, c)$
\If{$\mathsf{isLeftMaximal}(u)$}
\State $\mathsf{insert}(S, (u, d+1))$
\EndIf

\EndFor
\EndWhile
\end{algorithmic}
\label{algo:st}
\end{algorithm}

\begin{figure}[!h]
\centering
\includegraphics[width=0.45\textwidth]{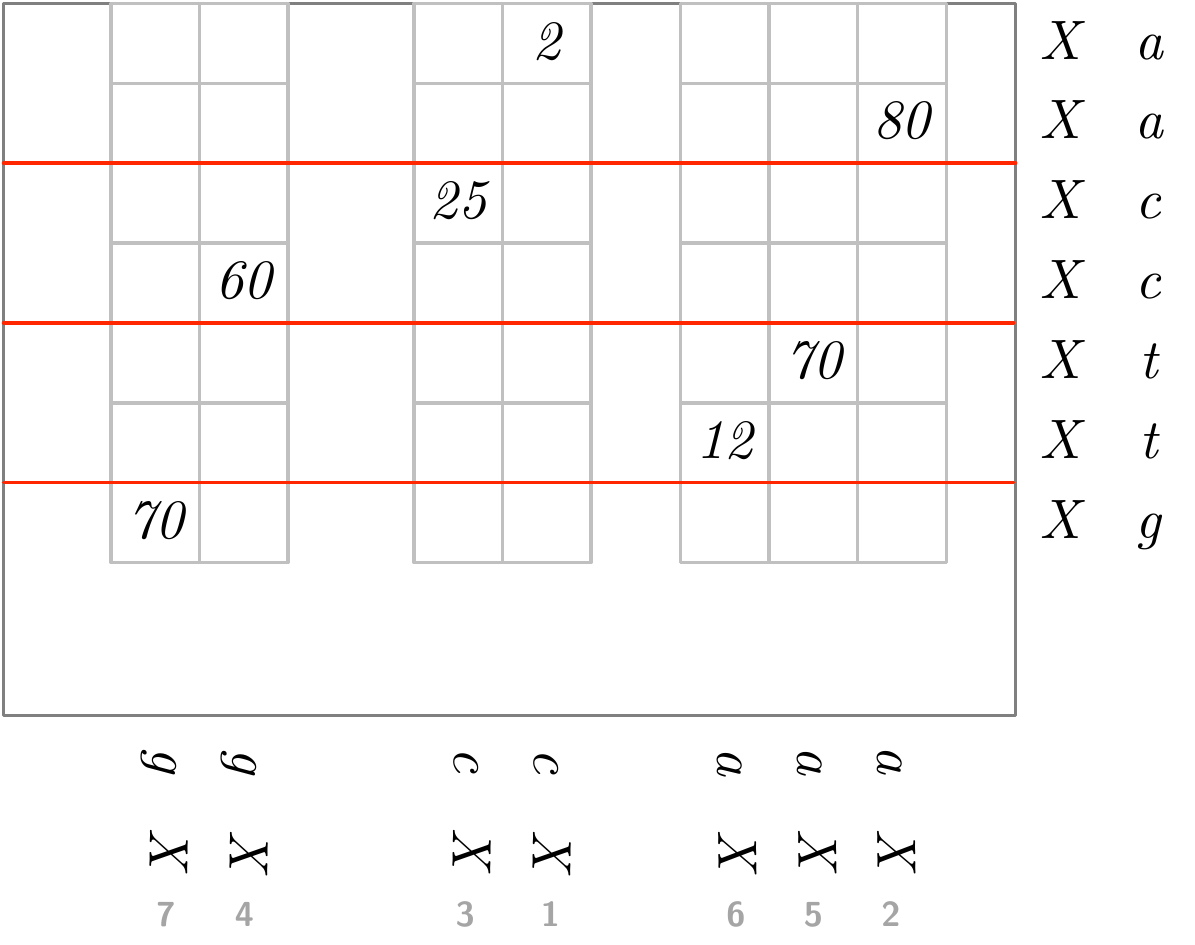}
\caption{
Reporting MEMs from an internal node $v$ labeled $\mathsf{label}(v)=X$ using the grid $\mathcal{G}$. The rows are labeled with the suffixes prefixed by $X$, while the column are labeled with the suffixes prefixed with the labels of $v$'s Weiner links. The horizontal red lines represents the partition of the $S\!A$ range for $X$ induced by the children of $v$. The grey numbers below the column labels are the $\mathsf{LF}^{-1}$ values. For each column $j'$, its associated $S\!A$ value is in the row $\mathsf{LF}^{-1}(j')=j$.
}
\label{fig:grid}
\end{figure}

\begin{algorithm}[t]
\caption{Report all-vs-all MEMs from a suffix tree node $v$. Arrays $BWT$ and $H$ for $\mathcal{R}^{h}$ are implicit in the pseudo-code. Node $v$ is encoded as described in Algorithm~\ref{algo:st}}\label{alg:mem}
\begin{algorithmic}[1]
\Require An internal node $v \in T$ with more than one Weiner link, run-length excess threshold $e$, and $d=|\mathsf{label}(v)|$. 
\Ensure List of MEMs among strings of $\mathcal{R}^{h}$ that can be computed from $v$.
\Procedure{$\mathsf{repMEM}$}{$v$, $d$, $e$}
\State $\mathcal{C} \gets \emptyset$ 
\For{$c \in \mathsf{enumerateRight}(v)$}\Comment{Partition $S\!A[v.i, v.j]$ according $v$'s children} 
\State $\mathcal{C} \gets \mathcal{C} \cup \{\mathsf{extendRight}(v, c)\}$
\EndFor

\State $\mathcal{M} \gets \emptyset$ 
\For{$x \gets 1$ to $|\mathcal{C}|$} \Comment{Get Weiner links for every child of $v$} 
\State $\mathcal{I}_{x} \gets \emptyset$
\For{$d \gets \mathsf{enumerateLeft}(\mathcal{C}[x])$}
\State $\mathcal{I}_{x} \gets \mathcal{I}_{x} \cup \{\mathsf{extendLeft}(\mathcal{C}[x], d)\}$
\EndFor
\State $\mathcal{M} \gets \mathcal{M} \cup \mathcal{I}_{x}$
\EndFor

\For{$\mathcal{I}_{a},\mathcal{I}_{b} \in \mathcal{M}$ with $\mathcal{I}_{a} \neq \mathcal{I}_{b}$} \Comment{$\mathcal{I}_{a}$ and $\mathcal{I}_{b}$ come from different children of $v$}

\For{$(X,Y) \in \mathcal{I}_{a} \times \mathcal{I}_{b}$}
\If{$X$ and $Y$ belong to distinct $S\!A$ bucket}
\For{$(q,e) \in X \times Y$}
\State $R_{q} \gets$ string in $\mathcal{R}^{h}$ for $SA[q]+1$
\State $R_{e} \gets$ string in $\mathcal{R}^{h}$ for $SA[e]+1$
\State $e' \gets \mathsf{rlExcess}(S\!A[q]+1, S\!A[e]+1, d)$ 
\If{$e'\leq e$}
\State Report MEM in $(q,e)$ 
\EndIf
\EndFor
\EndIf
\EndFor

\EndFor
\EndProcedure
\end{algorithmic}
\label{alg:cap}
\end{algorithm}

\end{document}